\documentclass[11pt,letterpaper]{article}

\usepackage[margin=1in]{geometry}

\usepackage[usenames,dvipsnames]{xcolor}

\usepackage{lmodern}
\usepackage[T1]{fontenc}
\usepackage[utf8]{inputenc}
\usepackage[tbtags]{amsmath}
\allowdisplaybreaks
\usepackage{amssymb,amsthm,mathtools}
\usepackage{thm-restate}
\usepackage{hyperref}
\usepackage[capitalise]{cleveref}
\definecolor{links}{RGB}{11, 85, 255}
\definecolor{cites}{RGB}{0, 200, 0}
\definecolor{urls}{RGB}{255, 116, 0}
\usepackage[numbers,compress]{natbib}
\usepackage{doi}
\usepackage{tikz} 
\usepackage{pgfplots}
\pgfplotsset{compat=1.14}
\usepgfplotslibrary{fillbetween}
\usepackage{hyphenat}
\usepackage[font=footnotesize]{caption}
\usepackage{subcaption}
\usepackage{appendix}
\usepackage{todonotes}
\usepackage{physics}
\usepackage[ruled]{algorithm2e} 
\usepackage{dsfont}
\usepackage{enumitem}
\usepackage{nicefrac}

\newcommand{\cA}{\mathcal{A}}
\newcommand{\cE}{\mathcal{E}}

\newcommand{\cF}{\mathcal{F}}
\newcommand{\cI}{\mathcal{I}}

\newcommand{\cM}{\mathcal{M}}

\newcommand{\R}{\mathbb{R}}
\newcommand{\N}{\mathbb{N}}

\newcommand{\poly}{\textsf{poly}}
\newcommand{\topk}{{\textsf{Top-}k}\xspace}

\newcommand{\alg}{\textsf{Alg}}
\newcommand{\adap}{\textsf{Adap}}
\newcommand{\Supp}{\textsf{Supp}}
\newcommand{\na}{\textsf{NA}}

\newcommand{\E}{\mathbb{E}}
\newcommand{\ind}{\mathds{1}}
\newcommand{\ones}{\mathbf{1}}
\newcommand{\up}[1]{^{(#1)}}

\newcommand{\ALG}{\textsf{ALG}}
\newcommand{\OPT}{\textsf{OPT}}
\newcommand{\cost}{\textsf{cost}}

\newcommand{\SD}{\textsf{SD}}
\newcommand{\LD}{\textsf{LD}}

\newcommand{\ba}{\boldsymbol{a}}

\newcommand{\bd}{\boldsymbol{d}}

\newcommand{\IGNORE}[1]{}

\DeclarePairedDelimiter\ceil{\lceil}{\rceil}
\DeclarePairedDelimiter\floor{\lfloor}{\rfloor}
\DeclarePairedDelimiter\inprod{\langle}{\rangle}

\newtheorem{theorem}{Theorem}[section]
\newtheorem{claim}[theorem]{Claim}

\newtheorem{lemma}[theorem]{Lemma}
\newtheorem{corollary}[theorem]{Corollary}
\newtheorem{conjecture}[theorem]{Conjecture}
\newtheorem{remark}[theorem]{Remark}
\newtheorem{definition}[theorem]{Definition}

\newtheorem{observation}[theorem]{Observation}

\begin{document}

\title{Submodular Norms with Applications To
Online Facility Location and Stochastic Probing}
\author{
Kalen Patton\footnote{(kpatton33@gatech.edu) School of Math Georgia Tech} \and
Matteo Russo\footnote{(mrusso@diag.uniroma1.it) DIAG, Sapienza Università di Roma}
\and
Sahil Singla\footnote{(ssingla@gatech.edu) School of Computer Science, Georgia Tech}}
\date{}

\maketitle

\begin{abstract}
\bigskip

Optimization problems often involve vector norms, which has led to extensive research on developing algorithms that can handle objectives beyond the $\ell_p$ norms.
Our work introduces the concept of \emph{submodular norms}, which are a versatile type of norms that possess marginal properties similar to submodular set functions. We show that submodular norms can accurately represent or approximate well-known classes of norms, such as $\ell_p$ norms, ordered norms, and symmetric norms. Furthermore, we establish that submodular norms can be applied to optimization problems such as online facility location, stochastic probing, and generalized load balancing. This allows us to develop a logarithmic-competitive algorithm for online facility location with symmetric norms, to prove a logarithmic adaptivity gap for stochastic probing with symmetric norms, and to give an alternative poly-logarithmic approximation algorithm for generalized load balancing with outer $\ell_1$ norm and inner symmetric norms.

\end{abstract}

 \setcounter{tocdepth}{1}
 {\small
    \tableofcontents
 }


\clearpage

\section{Introduction}

In the field of combinatorial optimization, norm objectives are frequently encountered. Canonical problems, such as the min-weight spanning tree and the $k$-median, involve searching for a feasible solution that minimizes the sum of costs, which is equivalent to the $\ell_1$ norm of the edge cost vector. On the other hand, canonical problems like bottleneck spanning tree and $k$-center aim to minimize the maximum of costs, which is equivalent to the $\ell_\infty$ norm of the edge cost vector. However, because $\ell_1$ and $\ell_\infty$ norms only capture the extreme Utilitarian and Egalitarian objectives respectively, significant research has been devoted to developing combinatorial optimization algorithms for more general norms  (see references in \Cref{sec:furtherRelated}). Among the commonly studied norms are $\ell_p$ norms, ordered norms, Orlicz norms, symmetric norms, and arbitrary monotone norms.

 Over the past decade, there is also a lot of effort towards designing online and stochastic algorithms for more general norms. For instance, remarkable progress has been made in developing algorithms beyond $\ell_p$ norms for various problems, such as load balancing \cite{ChakrabartyS19a, ChakrabartyS19b,IbrahimpurS20,IbrahimpurS21,IbrahimpurS22,Kesselheim0023}, set cover \cite{AzarBCCCG0KNNP16, NagarajanS17}, spanning trees \cite{Ibrahimpur-Thesis22}, and bandits with knapsacks \cite{Kesselheim020, Kesselheim0023}. Notably, most of the recent progress is for the class of symmetric norms, i.e., monotone norms that remain unchanged upon permutation of coordinates. This progress is partly due to Ky Fan's Dominance Theorem (refer to \cite{Bhatia97}), which reduces the problem of designing algorithms for symmetric norms to ordered norms (see \Cref{sec:defs} for a formal definition). Ordered norms are comparatively more manageable due to their explicit form. Given this progress on some combinatorial problems for symmetric norms, a natural question arises:

\begin{quote}
    \emph{What general norms and what combinatorial problems admit algorithms with good performance guarantees?}
\end{quote}

A challenge in making progress beyond symmetric norms is that such norms are not explicit, e.g., they may not be well approximated by ordered norms. 
In this work we introduce the class of \emph{submodular norms}, which is a broad class of norms with marginals properties  mimicking submodular set functions. We show that submodular norms either capture or approximate  popular classes of norms like $\ell_p$ norms, ordered norms, and symmetric norms. Moreover,   submodular norms  are amenable to some of the optimization problems  that were previously intractable like online facility location and stochastic probing.


\subsection{Norms and Submodularity}\label{sec:defs}

We start with the definitions of monotone, symmetric, and ordered norms. We will be only interested in norms defined in the positive orthant.

\begin{definition}[Monotone Norm]
    A \emph{monotone norm} is a function $\|\cdot\|: \R^n_{+} \rightarrow \R_{+}$ and is defined as $$\|x\| := \sup_{\ba \in \cA} \sum_{i} a_ix_i,$$ i.e, by a max of non-negative linear functions over set $\cA$. 
This is equivalent to saying that $\|x\|\geq \|y\|$ whenever $x \geq y \geq 0$ coordinatewise (hence, the name \emph{monotone}).
\end{definition}

\begin{definition}[Symmetric Norm]
    A monotone norm $\|\cdot\|$ is a \emph{symmetric norm} if, for any vector $x \in \R^n_{+}$ and for all of its coordinate permutations $\pi: [n] \rightarrow [n]$, $\|x\| = \|(x_{\pi(i)})_{i \in [n]}\|$.
\end{definition}

We remark that any symmetric norm can be written as $\sup_{\ba \in \cA} \inprod{\ba, x^\downarrow}$, where $x^\downarrow$ represents the sorted (in descending order) vector $x$, and $\cA$ is a set on non-negative descending vectors. This follows from the fact that $\max_\pi \sum_i a_i x_{\pi(i)} = \inprod{\ba^\downarrow, x^\downarrow}$ for non-negative vectors $a, x$.
In the special case where $\cA$ is a singleton, we have ordered norms.

\begin{definition}[Ordered Norm]
    A monotone norm is an \emph{ordered norm} if it can be written as $ \|x\| = \sum_{i} a_ix_i^\downarrow$, where $a_1 \geq \ldots \geq a_n \geq 0$.
\end{definition}

\paragraph*{Submodular Norms} Submodular set functions and their applications to optimization have been extensively studied; see books  \cite{Schrijver-Book03,Fujishige-Book05}. Intuitively, they capture the notion of decreasing marginal gains.  
Although submodular functions were originally defined for discrete settings, the notion has been generalized to arbitrary lattices, in particular to real vectors \cite{Bach18}. This leads to the following notion of continuous submodularity, which has found several applications in machine learning \cite{BBK-arXiv20,Bach13} and will be crucial in our definition of submodular norms.
We discuss standard properties of continuous submodularity in \Cref{sec:contSubmod}.

\begin{definition}[Continuous Submodularity] A real-valued function $f : \R^n_+ \to \R_+$ is \emph{continuously submodular} if for all $x, y \in \R^n_+$, we have
$
f(x \vee y) + f(x \wedge y) \leq f(x) + f(y),
$
where $x \vee y$ and $x \wedge y$ are the coordinate-wise max and min of $x$ and $y$ respectively. 
\end{definition}

Our first contribution is to define the following natural class of submodular norms.

\begin{definition}[Submodular Norm]
A monotone norm $\|\cdot\|$ is \emph{submodular} if it is continuously submodular.
\end{definition}

Examples of submodular norms include all  $\ell_p$ norms and Ordered norms (see \Cref{obs:Examples}). Moreover, the following theorem proved in \Cref{sec:approxSymmetric} shows that any symmetric norm can be approximated by a submodular norm.

\begin{theorem} \label{thm:symTosubmod}
Any symmetric norm can be $O(\log \rho)$ approximated by a submodular norm, where $\rho := \frac{\|(1,1,\ldots, 1)\|}{ \|(1,0,0,\ldots,0)\|} \leq n$. This approximation factor is tight up to $O(\log\!\log \rho)$ terms. 
\end{theorem}

There is an intimate connection between submodular norms and submodular set functions. 
Given a submodular norm $\|\cdot\|$,   the set function $f: 2^{[n]} \to \R^+$ by $f(S) := \|\ones_S\|$ is submodular, so every submodular norm is an extension of a submodular function. Moreover, if $f$ is a monotone submodular function with $f(\emptyset) = 0$, then $f$ can be extended to a submodular norm $\|\cdot\|$ by the \emph{Lov\'asz extension}:  
\[
\|x\| := \int_0^\infty f(\{j : t \leq x_j\}) dt.
\]
  This observation that every submodular set function induces a continuously submodular norm via its Lov\'asz extension   has appeared several times before \cite{Bach-NIPS10,Bach18}. However, our definition of submodular norms can capture many more natural norms.  E.g., all $\ell_p$ norms are submodular but  for $1<p<\infty$ they cannot be written as a Lov\'asz extension of a submodular set function since the dual-norm unit ball has an infinite number of vertices.

\begin{remark} \label{remark:DRsubmod}
A commonly studied variant of continuous submodularity is  DR-submodularity \cite{BianB019, FeldmanK20, NiazadehRW18}: a function $f : \R^d_+ \to \R_+$ is \emph{DR-submodular} if it satisfies \emph{diminishing returns} meaning $ f(w + ae_i) - f(w) \leq f(x + ae_i) - f(x)$ for all $x, w \in \R^d_+$ with $x \leq w$, $i \in [d]$, and $a \geq 0$. 
 It is known that continuous submodularity is  equivalent to having  this diminishing returns inequality only when  $x_i = w_i$; hence continuous submodularity is a weaker property. 
The class of DR submodular functions turns out to be uninteresting when looking at norms since  the only DR-submodular norm is the $\ell_1$-norm (up to rescaling coordinate-wise). See \Cref{sec:contSubmod} for proofs. 
\end{remark}

\subsection{Applications} \label{sec:introApplications}
In addition to being a natural class of norms, submodular norms find multiple applications. We will explore two specific applications-one in the domain of online algorithms and another in the field of stochastic optimization.

\paragraph*{Online Facility Location} In this problem we are given a metric space $(\cM, d)$ equipped with metric $d: \cM \times \cM \rightarrow \R_{\geq 0}$, along with a cost function $f : \cM \to \R_+$ and a norm $\| \cdot \|: \R_+^{n} \rightarrow \R_+$. At each time step $i \in [n]$, an adversary produces a new request $x_i \in \cM$ and the algorithm decides to either assign $x_i$ to the 
closest already-open facility in the set  $F_{i-1}$, thereby incurring a connection cost $d(x_i, F_{i-1})$, or to open a new facility $q$ and assign request $x_i$ to facility $q$, thereby incurring a connection cost $d(x_i, q)$ and an opening cost $f(q)$.
Let $F$ be the final set of opened facilities, let $F_i$ be the set of facilities opened until (and including) the $i$-th request, and let $\bd = (d_1, \dots, d_n) \in \R^{n}_+$ be the vector of connection costs $d_i := d(x_i, F_i)$. Our goal is to minimize the total cost  $ \sum_{q \in F} f(q) + \|\bd\|.$

Online facility location was introduced by Meyerson for $\ell_1$ norm \cite{Mey01}, where he showed an $O(\log n)$ competitive algorithm. A tight competitive ratio of $\Theta(\nicefrac{\log n}{\log\!\log n})$ was later obtained by Fotakis \cite{Fotakis-Journal08}. When all requests are given up front (offline setting), it is a classical NP-hard problem where we can design $O(1)$ approximation algorithm, even for general norms \cite{GMS-arXiv22}. In the online setting, however, no non-trivial algorithm was previously known beyond $\ell_1$ norms.

\begin{restatable}{theorem}{nonUnif}\label{thm:non-unif}
    For online facility location problem with a submodular norm $\|\cdot\|$, there exists a randomized online algorithm that obtains cost at most $O(\log \rho) \cdot \sum_{z \in F^*} f(z) + O(1) \cdot \|\bd^*\|$, where $F^*$ and $\bd^*$ are the set of facilities and vector of assignment distances respectively given by the offline optimum algorithm and $\rho := \frac{\|(1,1,\ldots, 1)\|}{ \min_i \|e_i\|} \leq n \cdot \frac{\max_i \|e_i\|}{\min_i \|e_i\|}$. 
\end{restatable}

Since any symmetric norm can be $O(\log \rho)$ approximated by a submodular norm by \Cref{thm:symTosubmod}, we get the following corollary.

\begin{corollary}
For online facility location problem with a symmetric norm, there exists an $O(\log \rho)$-competitive randomized algorithm.
\end{corollary}

For concreteness, this corollary implies an $O(\log n)$-competitive algorithm for $\ell_1$ norm, which matches Meyerson's bound \cite{Mey01}, an $O(1)$-competitive algorithm for $\ell_\infty$ norm, and an $O(\log k)$-competitive algorithm for $\topk$ norm. This is tight up to an $O(\log \log \rho)$ factor for any symmetric norm by the lower bound construction given in \Cref{thm:ofl-lower-bound}.

The proof of \cref{thm:non-unif} relies on generalizing Meyerson's algorithm beyond $\ell_1$ norms. 
Meyerson's algorithm constructs a new facility at each demand point $x_i$ with probability ${d(x_i, F_{i-1})}/{f}$, thereby balancing the cost of assigning the demand against the cost of constructing a new facility. A natural generalization of this algorithm to general norms is to construct a new facility with probability ${\delta_i}/{f}$, where marginal cost $\delta_i = \|(d_1, \dots, d_{i-1}, d(x_i, F_{i-1}), 0, \dots, 0)\| - \|\bd_{\leq i-1}\|.$ Unfortunately, we will show that such an algorithm is   $\Omega(n)$-competitive even for the $\ell_\infty$ norm. Our crucial change to Meyerson's algorithm is to carefully define \emph{auxiliary assignment costs} $\hat{d_i}$ which upper bound the true costs $d_i$. Now we use $\hat{d}_i$ instead of $d_i$ to calculate the marginal cost $\delta_i$. Due to norm submodularity, this underestimates the marginal costs, making the algorithm more inclined to assign demand points instead of constructing new facilities.

Next, we discuss a stochastic optimization application of submodular norms. 

\paragraph*{Stochastic Probing} This problem is a natural stochastic generalization of constrained submodular maximization.  Here, we are given probability distributions of $n$ independent random variables $X = (X_1, \dots, X_n) \in \R_+$, a downward-closed  set family $\cF \subseteq 2^{[n]}$, and a monotone objective $f : \R^n_+ \to \R_+$. The goal is to select a feasible set $S \in \cF$ of variables in order to maximize $f(X_S)$. 
The optimal strategy for this problem is generally \emph{adaptive}, i.e., it selects elements of $S$ one at a time and may change its decisions based on observations of the selected variables.

Since adaptive strategies are complicated (could be an exponential-sized decision tree) and hard to implement for many applications of stochastic probing, we are interested in finding non-adaptive algorithms that maximize $\max_{S \in \cF} \E[f(X_S)]$. The main question, which has been studied in several papers \cite{AN16,GN-IPCO13,GNS-SODA16,GNS-SODA17,BSZ-Random19,EKM-COLT21}, 
is how much do we lose when we move from adaptive to non-adaptive algorithms, i.e., if $\adap(X, \cF, f)$ denotes the optimal adaptive strategy and $\na(X, \cF, f)$ denotes the optimal non-adaptive algorithm, then what is the maximum possible   \emph{adaptivity gap}  $\frac{\adap(X, \cF, f)}{\na(X, \cF, f)}$.

For submodular set functions, the worst-case adaptivity gap is known to be $2$ \cite{GNS-SODA17,BSZ-Random19}. An interesting conjecture posed in \cite{GNS-SODA17} is whether the adaptivity gap for \emph{XOS} set functions is poly-logarithmic in $n$, where an XOS set function $f: 2^{[n]} \rightarrow \R_+$ is a max over linear set functions.  Since a monotone norm is nothing but a max over linear functions (given by the dual-norm unit ball), they form an extension of XOS set functions from the hypercube to all non-negative real vectors. Thus, we can generalize the conjecture of \cite{GNS-SODA17} to the following:

\begin{conjecture}
The adaptivity gap  for stochastic probing with monotone norms is poly-$\log(n)$.
\end{conjecture}

Although we are not able to resolve this general conjecture, we make progress by resolving it for all symmetric norms.

\begin{restatable}{theorem}{adapGapSymmetric}\label{thm:adapGapSymmetric}
The adaptivity gap for stochastic probing with symmetric norms is $O(\log n)$.
\end{restatable}

The proof of this result  relies on first  approximating the symmetric norm by a submodular norm as given in \Cref{thm:symTosubmod}. Next, we generalize the technique of bounding adaptivity gaps for submodular set functions in \cite{BSZ-Random19} to submodular norms.

\paragraph*{Generalized Load Balancing}
The setting of generalized load balancing was introduced by Deng, Li, and Rabani \cite{DengLR023} as a way to capture many related make-span minimization problems. In this problem, we have $n$ jobs and $m$ unrelated machines, and we seek to find an assignment $\sigma : [n] \to [m]$ of jobs to machines. Each job $j$ has a processing time $\mathbf p_{ij} \in \R_+$ on machine $i$. Each machine also has a monotone  \emph{inner norm} $\psi_i : \R_+^n \to \R_+$, and the \emph{load} of machine $i$ is given by 
$$load_i(\sigma) := \psi_i\Big[(\mathbf p_{ij} \cdot \ind_{\sigma(j) = i})_{j \in [n]}\Big].$$
Additionally, the costs over all machines are aggregated with a monotone \emph{outer norm} $\phi : \R_+^m \to \R_+$. The goal is to find an assignment of jobs to machines $\sigma$ that minimizes $\phi[(load_i(\sigma))_{i \in [n]}]$.

In \cite{DengLR023}, the authors study the setting where $\phi$ and all $\psi_i$ are symmetric, giving an $O(\log n)$ approximate algorithm using LP based methods. Additionally, when $\phi$ is the $\ell_1$ norm, and $load_i(\sigma)$ is a monotone submodular function of $\sigma^{-1}(i)$, it is known that a simple greedy approach can get a $O(\log n)$ approximation as well \cite{SvitkinaZT10}.

Viewing this problem under the framework of submodular norms, we note that if the $\psi_i$ are submodular norms, then $\psi_i\Big[(\mathbf p_{ij} \cdot \ind_{\sigma(j) = i})_{j \in [n]}\Big]$ is indeed a monotone submodular function of $\sigma^{-1}(i)$. This means that in the setting where $\phi$ is an $\ell_1$ norm, and $\psi_i$ are submodular norms, the result of \cite{SvitkinaZT10} implies a $O(\log n)$ approximation. Combining this result with approximation of symmetric norms by submodular norms from \cref{lem:log-rho-approx}, we also note that a $O(\log^2 n)$ approximate algorithm can be obtained in the setting where $\phi$ is $\ell_1$, and the $\psi_i$ are symmetric norms. This gives an alternative way to achieve the $\poly \log (n)$ approximation factor that \cite{DengLR023} gets for this setting, using the greedy algorithm of \cite{SvitkinaZT10} instead of LP methods.

\IGNORE{\color{red}
\paragraph{Vertex Cover.} 
In vertex cover with submodular norms, we are given a graph $G= (V,E)$ and a
norm $\|\cdot \| : \R_+^{|V|} \rightarrow \R_+$. The goal is to find a  subset of vertices $S$ s.t. every edge is incident to a vertex in $S$ while minimizing $\|\ind_S\|$. This problem generalizes  the classical NP Hard problem of vertex cover where we want to maximize the sum of the weights of vertices in $S$. 
There is a simple $2$ approximation algorithm for the classical problem, which is known to be tight under the unique games conjecture \cite{WS-Book,KR-JCSS08}.  We ask for what norms can we obtain good approximation algorithms for vertex cover.

\begin{theorem}
There exists a $2$-approximation algorithm for vertex cover with submodular norms. 
\end{theorem}
An immediate corollary of this theorem is an $O(\log \rho)$-approximation algorithm for vertex cover with symmetric norms due to \Cref{thm:symTosubmod}.  
Our technique in the proof of this theorem is to generalize the result of \cite{IwataN09} from submodular set functions to submodular norms.
}

\subsection{Further Related Work} \label{sec:furtherRelated}

 In recent years, there has been a surge of interest in the study of general norms. Some of the combinatorial problems that have been studied beyond $\ell_p$ norms are  load balancing \cite{ChakrabartyS19a,ChakrabartyS19b,IbrahimpurS20,IbrahimpurS21}, $k$-clustering \cite{BSS-STOC18,ChakrabartyS19a}, vector scheduling \cite{Kesselheim0023,DLR-SODA23,IbrahimpurS22},  set cover \cite{AzarBCCCG0KNNP16, NagarajanS17}, spanning trees \cite{Ibrahimpur-Thesis22}, and generalized assignment with convex costs \cite{GuptaKP12,Kesselheim0023}. Beyond combinatorial optimization, general norms have been recently studied for problems such as mean estimation with statistical queries \cite{LiNRW19}, nearest-neighbor search \cite{AndoniNNRW17, AndoniNNRW18},   regression  \cite{AndoniLSZZ18, SongWYZZ19}, and communication complexity~\cite{AndoniBF23}.

Continuous submodular functions have been extensively studied in the machine learning literature. We refer to the beautiful article of Bach \cite{Bach18} for their properties. Some of their applications to combinatorial optimization are discussed in \cite{AxelrodLS20, NiazadehRW18} and to learning are discussed in \cite{ZhangDCH022, FeldmanK20}. The fact that submodular set function induces a norm via its Lov\'asz extension has found several applications for regression since they induce sparsity \cite{Bach-NIPS10,Bach13}.

\paragraph*{Paper Outline.} 
Our work revolves around submodular norms and combinatorial optimization problems where the objective function is a submodular norm. In \Cref{sec:submod}, we illustrate the key properties of submodular norms  and the extent to which they serve as a good proxy for other classes of norms. In this respect, we identify a crucial parameter $\rho$ that structurally characterizes a given submodular norm and may be of independent interest. In \Cref{sec:online-fac}, we leverage these properties to derive a competitive algorithm for online facility location. Finally, in \Cref{sec:stochProbing}, we provide an application of submodular norms to adaptivity gaps for stochastic probing.
\section{Submodular Norms}\label{sec:submod}

We study  properties of submodular norms and how they relate to  other commonly studied norms.

\subsection{Properties and Important Special Cases}\label{sec:submod-relation}

We first discuss some common examples of submodular norms.
 
\begin{observation} \label{obs:Examples}
The following norms are submodular:  
\begin{enumerate} [noitemsep,topsep=0pt]
    \item All $\ell_p$ norms are submodular.
    \item All $\topk$ and ordered norms are submodular.
    \item For a matroid $\cM = ([n], \cI)$, the \emph{matroid rank norm} $\|x\| := \max_{S \in \cI} (\sum_{i \in S} x_i)$ is submodular.
\end{enumerate}
\end{observation}
\begin{proof}
    To see that $\ell_p$ norms are submodular, it suffices to show that for any monotone concave $g : \R_+ \to \R_+$, and submodular $f : \R_+^n \to \R_+$, the function $g \circ f$ is submodular. We can then apply this when $f = \|x\|_p^p$ and $g(y) = y^{1/p}$.
    
    To prove the claim, notice that for $x, y \in \R_+^n$,
    \begin{align*}
    g(f(x \vee y)) - g(f(x)) \leq g(f(x \vee y) - f(x) + f(x \wedge y)) - g(f(x \wedge y))
    \leq g(f(y)) - g(f(x \wedge y)).
    \end{align*}
    
    On the other hand, $\topk$ norms and matroid rank norms are special cases of Lov\'asz extensions. The matroid rank norm is the Lov\'asz extension of the rank function, and a $\topk$ norm is a matroid rank norm for the $k$-uniform matroid. 
\end{proof}

Submodular norms are also closed under several natural operations.

\begin{lemma} \label{lem:propSubmodNorms}
    The following operations return a submodular norm:
    \begin{enumerate}[noitemsep,topsep=0pt]
        \item Any rescaling of the coordinates of a submodular norm.
        
        \item Sums of partial\footnote{Partial means norms defined on a subset of coordinates with every other coordinate treated as $0$.} submodular norms.
        
        \item Any conical combination of submodular norms  is submodular.\footnote{Let $x_1, \ldots, x_m \in \R^n$ be real-valued vectors. We say that $y = \sum_{i \in [m]} \alpha_ix_i$ is a \emph{conical} combination of the vectors if $\alpha_i \geq 0$ for all $i \in [m]$.}  
        
    \end{enumerate}
\end{lemma}

\begin{proof}
    The first property follows since coordinate-wise rescaling of vectors commutes with coordinate-wise max and min.

    The second property follows from the fact that a partial submodular norm is a submodular semi-norm (i.e., a norm without the requirement to be positive definite). A sum of semi-norms remains a semi-norm, and from \cite{Bach18}, a sum of continuously submodular functions is continuously submodular.
    
    Finally, it is folklore that conical combinations of norms are norms, and it is also easy to show that such combinations also preserve continuous submodularity (see \cite{Bach18}).
\end{proof}

Besides their strict containment of many common norms, submodular norms are also powerful because they can be used to approximate other norms. In \cref{sec:approxSymmetric}, we will discuss how symmetric norms can be approximated by submodular norms up to logarithmic factors. In addition, we note in \cref{sec:beyond-sym} that submodular norms may approximate a much larger class of norms than just symmetric, although they are still far from the most general class of monotone norms. These approximation relations are summarized in \cref{fig:norms}.
\begin{figure}[h]
    \centering

\tikzset{every picture/.style={line width=0.75pt}} 

\begin{tikzpicture}[x=0.75pt,y=0.75pt,yscale=-1,xscale=1]

\draw  [line width=1.5]  (0,56) .. controls (0,25.07) and (25.07,0) .. (56,0) -- (414,0) .. controls (444.93,0) and (470,25.07) .. (470,56) -- (470,224) .. controls (470,254.93) and (444.93,280) .. (414,280) -- (56,280) .. controls (25.07,280) and (0,254.93) .. (0,224) -- cycle ;
\draw  [line width=1.5]  (240,169.25) .. controls (240,158.62) and (248.62,150) .. (259.25,150) -- (430.75,150) .. controls (441.38,150) and (450,158.62) .. (450,169.25) -- (450,240.75) .. controls (450,251.38) and (441.38,260) .. (430.75,260) -- (259.25,260) .. controls (248.62,260) and (240,251.38) .. (240,240.75) -- cycle ;
\draw  [line width=1.5]  (30,142) .. controls (30,124.33) and (44.33,110) .. (62,110) -- (358,110) .. controls (375.67,110) and (390,124.33) .. (390,142) -- (390,238) .. controls (390,255.67) and (375.67,270) .. (358,270) -- (62,270) .. controls (44.33,270) and (30,255.67) .. (30,238) -- cycle ;
\draw    (60,12) -- (60,99) ;
\draw [shift={(60,101)}, rotate = 270] [color={rgb, 255:red, 0; green, 0; blue, 0 }  ][line width=0.75]    (10.93,-4.9) .. controls (6.95,-2.3) and (3.31,-0.67) .. (0,0) .. controls (3.31,0.67) and (6.95,2.3) .. (10.93,4.9)   ;
\draw [shift={(60,10)}, rotate = 90] [color={rgb, 255:red, 0; green, 0; blue, 0 }  ][line width=0.75]    (10.93,-4.9) .. controls (6.95,-2.3) and (3.31,-0.67) .. (0,0) .. controls (3.31,0.67) and (6.95,2.3) .. (10.93,4.9)   ;
\draw    (228,240) -- (42,240) ;
\draw [shift={(40,240)}, rotate = 360] [color={rgb, 255:red, 0; green, 0; blue, 0 }  ][line width=0.75]    (10.93,-4.9) .. controls (6.95,-2.3) and (3.31,-0.67) .. (0,0) .. controls (3.31,0.67) and (6.95,2.3) .. (10.93,4.9)   ;
\draw [shift={(230,240)}, rotate = 180] [color={rgb, 255:red, 0; green, 0; blue, 0 }  ][line width=0.75]    (10.93,-4.9) .. controls (6.95,-2.3) and (3.31,-0.67) .. (0,0) .. controls (3.31,0.67) and (6.95,2.3) .. (10.93,4.9)   ;
\draw [line width=0.75]    (445.2,210) -- (420,210) -- (399.2,210) ;
\draw [shift={(397.2,210)}, rotate = 360] [color={rgb, 255:red, 0; green, 0; blue, 0 }  ][line width=0.75]    (10.93,-3.29) .. controls (6.95,-1.4) and (3.31,-0.3) .. (0,0) .. controls (3.31,0.3) and (6.95,1.4) .. (10.93,3.29)   ;
\draw [shift={(447.2,210)}, rotate = 180] [color={rgb, 255:red, 0; green, 0; blue, 0 }  ][line width=0.75]    (10.93,-3.29) .. controls (6.95,-1.4) and (3.31,-0.3) .. (0,0) .. controls (3.31,0.3) and (6.95,1.4) .. (10.93,3.29)   ;
\draw    (420,12) -- (420,138) ;
\draw [shift={(420,140)}, rotate = 270] [color={rgb, 255:red, 0; green, 0; blue, 0 }  ][line width=0.75]    (10.93,-4.9) .. controls (6.95,-2.3) and (3.31,-0.67) .. (0,0) .. controls (3.31,0.67) and (6.95,2.3) .. (10.93,4.9)   ;
\draw [shift={(420,10)}, rotate = 90] [color={rgb, 255:red, 0; green, 0; blue, 0 }  ][line width=0.75]    (10.93,-4.9) .. controls (6.95,-2.3) and (3.31,-0.67) .. (0,0) .. controls (3.31,0.67) and (6.95,2.3) .. (10.93,4.9)   ;

\draw (181,11) node [anchor=north west][inner sep=0.75pt]   [align=left] {\underline{Monotone Norms}};
\draw (91,122) node [anchor=north west][inner sep=0.75pt]   [align=left] {\underline{Submodular Norms}};
\draw (261,161) node [anchor=north west][inner sep=0.75pt]   [align=left] {\underline{Symmetric Norms}};
\draw (261,190) node [anchor=north west][inner sep=0.75pt]   [align=left] {$\displaystyle \ell _{p}$-norms};
\draw (261,221) node [anchor=north west][inner sep=0.75pt]   [align=left] {Ordered norms};
\draw (61,161) node [anchor=north west][inner sep=0.75pt]   [align=left] {Lovász extensions};
\draw (30,41) node [anchor=north west][inner sep=0.75pt]    {$\sqrt{n}$};
\draw (121,242.4) node [anchor=north west][inner sep=0.75pt]    {$n$};
\draw (395,222.4) node [anchor=north west][inner sep=0.75pt]    {$\approx \log n$};
\draw (59,190) node [anchor=north west][inner sep=0.75pt]   [align=left] {Sum of partial $\displaystyle \ell _{p}$-norms};
\draw (91,41) node [anchor=north west][inner sep=0.75pt]   [align=left] {Max of positive linear functionals};
\draw (427,52.4) node [anchor=north west][inner sep=0.75pt]    {$n$};

\end{tikzpicture}

    \caption{The containment relationships between monotone norms, submodular norms, and symmetric norms, along with some examples. The ``distances'' shown indicate the worst-case approximation factor (up to constants) for a norm in each outer class by a norm in the corresponding inner class.}
    \label{fig:norms}
\end{figure}
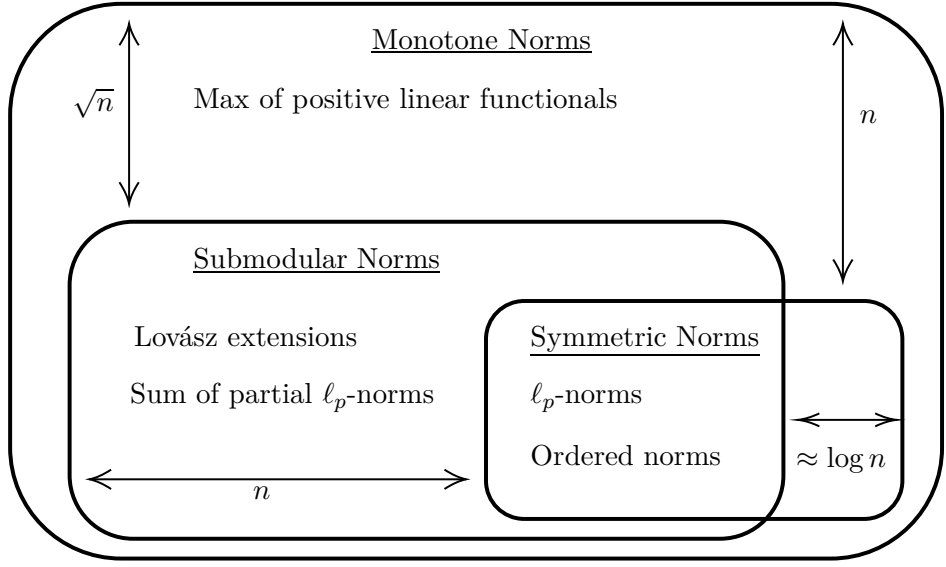

\subsection{Approximation of Symmetric Norms} \label{sec:approxSymmetric}
A major benefit of studying submodular norms is that they can approximate any symmetric norm. Indeed, previous works have noted that symmetric norms can be approximated by an ordered norm up to a factor of $O(\log n)$ \cite{ChakrabartyS19a,Kesselheim0023}. 
For our purposes, it will be useful to make this approximation  more precise by replacing $\log n$ with $\log \rho$, where parameter $\rho$ is defined as follows:
\begin{definition}
If $e_1, \dots, e_n$ denote the standard basis vector
and let $\ones_{\leq i} := \sum_{1\leq j \leq i} e_j$ denote the vector with $1$s at the first $i$ coordinates  and $0$ otherwise.
Then  for any monotone norm $\|\cdot\|$ we define the parameter 
    $$\rho_{\|\cdot\|} := \frac{\|\ones_{\leq n} \|}{\min_{i \in [n]} \|e_i\|}.$$
When the norm is clear from context, we  simply write $\rho = \rho_{\|\cdot\|}$.
\end{definition}
Notice that for symmetric norms, we have $\rho = \frac{\|\ones_{\leq n} \|}{\|e_1\|} \leq n$. One can think of $\rho$ for symmetric norms as a measure of how closely a norm behaves like $\|\cdot\|_1$ (large $\rho$) versus $\|\cdot\|_\infty$ (small $\rho$).

\begin{observation}
For $\ell_p$ norms, we have $\rho_{\|\cdot\|_p} = n^{1/p}$. For $\topk$ norms, we have $\rho_{\|\cdot\|_\topk} = k$.
\end{observation}

As we will see, the parameter $\rho$ appears again in both the upper and lower bound analysis in \cref{sec:online-fac} and \cref{sec:ofl-lower-bound}, making the improvement from $\log n$ to $\log \rho$ in \cref{lem:log-rho-approx} necessary for tight bounds in our applications.

\begin{lemma} \label{lem:log-rho-approx}
For any symmetric norm $\|\cdot\|$ with $\rho_{\|\cdot\|} = \rho$, there is an ordered norm $\|\cdot\|'$ such that $\|x\| \leq \|x\|' \leq 2(\log \rho + 1) \cdot \|x\|$.
\end{lemma}

\begin{proof}
Let $\|x\| = \max_{a \in \cA} \inprod{a, x^\downarrow}$. Without loss of generality, assume $\|e_1\| = 1$, so $ \|\ones_{\leq n}\| =\|(1,\dots, 1)\| = \rho$. 

Let $1 = m_0 \leq m_1 \leq \dots \leq m_{\floor{\log \rho}}$ be such that $m_j$ is the least integer with $\|\ones_{\leq m_j}\| \geq 2^j$. Let $a_0, \dots, a_{\floor{\log \rho}} \in \cA$ be such that $\|\ones_{\leq m_j}\| = \inprod{a_j, \ones_{\leq m_j}}$.

Now consider the ordered norm $\|x\|' := 2\inprod{a^*, x^\downarrow}$, where $a^* = \sum_j a_j$. Clearly, we have
$$
\frac{1}{2(\floor{\log \rho} + 1)}\|x\|' \leq \max_{j} \inprod{a_j, x^\downarrow} \leq \max_{a \in \cA} \inprod{a, x^\downarrow} = \|x\|.
$$
Additionally, notice that for any $x \in \R^n_+$, we can write $x^\downarrow = \sum_{k \in [n]} \lambda_k \ones_{\leq k}$ for some $\lambda_k \geq 0$. We have
$$
\|x\|' = 2\sum_j \sum_k \lambda_k \inprod{a_j, \ones_{\leq k}} 
\geq 2\sum_k \lambda_k \max_{j} \inprod{a_j, \ones_{\leq k}} 
\geq^{\dagger} \sum_k \lambda_k \|\ones_{\leq k}\| \geq \|x\|.
$$
Here, $\dagger$ follows from rounding each $k$ down to the nearest $m_j$ and using 
\[\inprod{a_j, \ones_{\leq k}} \geq \inprod{a_j, \ones_{\leq m_j}} = \|\ones_{\leq m_j}\| \geq \frac{1}{2}\|\ones_{\leq k}\|. \qedhere \]
\end{proof}

\begin{remark}
    The ordered norm approximation in \cref{lem:log-rho-approx} can also be obtained in polynomial time with only value oracle access to $\|\cdot \|$. To do this, we modify the above proof by instead writting $\|x\|' := \sum_j \inprod{b_j, x^\downarrow}$, where $b_j := \frac{\|\ones_{\leq m_j}\|}{m_j} \cdot \ones_{\leq m_j}$, since both $\|\ones_{\leq m_j}\|$ and $m_j$ may be obtained from a value oracle by binary search. The proof of the lemma with this modified norm can be found in \cref{sec:missingSubmod}.
\end{remark}

\paragraph*{Tightness of approximation}

In the worst case where $\rho = \Omega(n)$, the $\log n$ factor turns out to be nearly the best possible factor for approximation of a symmetric norm by a submodular norm. 
The following lemma shows that the construction in \Cref{lem:log-rho-approx} is tight up to $O(\log\!\log n)$ factors.
\begin{lemma}
    For any $\varepsilon \in (0, 1/2)$, define
    $$\|x\| := \max_{k \in [n]} k^{-\varepsilon} \cdot \inprod{ \ones_{\leq k}, x^\downarrow}.$$ 
    For any submodular norm $\|\cdot\|'$ such that $\|x\|' \geq \|x\|$ for all $x \in \R^n_+$, there exists $y \in \R^n_+$ such that $\|y\|' \geq C\frac{\varepsilon}{1-\varepsilon}(\log n)^{1-\varepsilon} \|y\|$. Taking $\varepsilon = \frac{1}{\log\log n}$ gives $\|y\|' \geq \Omega(\frac{\log n}{\log \log n})\|y\|$.
\end{lemma}

\begin{proof}
    Let $y$ be defined by $y_k := \frac{k^{\varepsilon} - (k-1)^{\varepsilon}}{\varepsilon}$. A simple calculation yields  that 
    $$k^{-(1-\varepsilon)} \leq y_k \leq (k-1)^{-(1-\varepsilon)}.$$
    Additionally, we have
    $$
    \|y\| = \max_{k \in [n]} k^{-\varepsilon}\cdot\sum_{i=1}^k y_i = \max_{k \in [n]} k^{-\varepsilon} \cdot \frac{k^\varepsilon}{\varepsilon} = \frac{1}{\varepsilon}.
    $$

    Now to estimate $\|y\|'$, first note that we may assume $\|\cdot\|'$ to be symmetric, otherwise replace $\|\cdot\|'$ with its average over all permutations of inputs. We will inductively show that for $j \in [\nicefrac{n}{\log n}]$, we have $\|y_{\leq j}\|' \geq b_j := \frac{(\log (j+1))^{1-\varepsilon}}{4(1-\varepsilon)}$.

    For $j = 1$, we check 
    $$b_1 \leq \frac{1}{4(1-\varepsilon)} \leq \frac{1}{2} \leq 1 = \|y_{\leq 1}\| \leq \|y_{\leq 1}\|'.$$
    Now assume the claim holds for a given $j \geq 1$. Consider $z \in \R^d_+$ defined by
    $$
    z_i = \begin{cases}
        y_i & 1 \leq i \leq j,\\
        y_{j+1} & j < i \leq j + \ell,\\
        0 & j + \ell < i,
    \end{cases}
    $$
    where $\ell := \ceil{(j+1)\log(j+1)}$. Notice that by submodularity and symmetry, we have
    $$
    \|y_{\leq j+1}\|' \geq \|y_{\leq j}\|' + \frac{\|z\|' - \|y_{\leq j}\|'}{\ell} \geq b_j + \frac{\|z\|' - b_j}{\ell}.
    $$

    We see that
    $$
    \|z\|' \geq \|z\| \geq (j+\ell)^{-\varepsilon} \cdot (j+\ell) y_{j+1} \geq \left(\frac{j+\ell}{j+1}\right)^{1-\varepsilon} \geq \log(j+1)^{1-\varepsilon}.
    $$
    Thus, we have $\|z\|' - b_j \geq \frac{1}{2}\log (j+1)^{1- \varepsilon}$. Finally, we have
    $$
    \|y_{\leq j+1}\|' \geq b_j + \frac{\log(j+1)^{1-\varepsilon}}{2\ell} \geq b_j + \frac{\log(j+1)^{-\varepsilon}}{4(j+1)} \geq b_j + \int_{j+1}^{j+2}\frac{(\log x)^{-\varepsilon}}{4x}dx = b_{j+1}. \qedhere
    $$
    
\end{proof}

\subsection{Beyond Symmetric Norms}\label{sec:beyond-sym}

Given that submodular norms allow us to approximate symmetric norms up to $\log n$ factors, we may ask if other classes of norms can be similarly approximated.
We  note that there exist submodular norms that are an $\Omega(n)$ factor away from any symmetric norm, which suggests that symmetric norms are not the largest class of norms which are approximated by submodular norms. Indeed, sums of partial $\ell_p$ or $\topk$ norms, such as those considered in \cite{NagarajanS17}, are submodular but can be highly asymmetric. Thus, although we focus on approximating symmetric norms by submodular norms, it is likely that many asymmetric norms admit submodular approximations. However, we leave the problem of characterizing these norms for future work.

In the following lemmas, we adopt the notation $x_S$ for $x \in \R^n$ and $S \subseteq [n]$ to denote $x$ after zeroing out all entries except those at indices in $S$, as well as $\ones_S$ to denote the indicator vector of $S$.

\begin{lemma}
    There exists a submodular norm $\|\cdot\|'$ for which any symmetric norm $\|\cdot\|$ satisfying $\|x\| \leq \|x\|'$ for all $x \in \R^n_+$, also has $\|y\|' \geq \Omega(n) \cdot \|y\|$ for some $y \in \R^n_+$.
\end{lemma}
\begin{proof}
    Let $A := \{1, \dots, \nicefrac{n}{2}\}$ and $B := \{\nicefrac{n}{2} + 1, \dots, n\}$. Define the norm $\|x\|':= \|x_A\|_\infty + \|x_B\|_1$. Notice that $\|\cdot\|'$ is a sum of partial $\ell_p$ norms, so it is submodular by \Cref{lem:propSubmodNorms}.
     Now suppose $\|\cdot\|$ is a symmetric norm with $\|x\| \leq \|x\|'$ for all $x \in \R^n_+$. Then $\|\ones_A\| \leq \|\ones_A\|' = 1$. However, we also have $\|\ones_B\| = \|\ones_A\|$ by symmetry, and $\|\ones_B\|' = \nicefrac{n}{2}$. Thus, taking $y = \ones_B$, we have our lemma.
\end{proof}
 
In the most general setting of monotone norms, however, submodular norms cannot give better than an $\Omega(\sqrt{n})$ approximation. The proof of this fact is similar to the canonical proof of the $\Omega(\sqrt{n})$ factor gap between submodular set functions and XOS set functions.

\begin{lemma}
    There exists a monotone norm $\|\cdot\|$ for which any submodular norm $\|\cdot\|'$ satisfying $\|x\| \leq \|x\|'$ for all $x \in \R^n_+$, also has $\|y\|' \geq \Omega(\sqrt{n}) \cdot \|y\|$ for some $y \in \R^n_+$.
\end{lemma}
\begin{proof}
    Partition $[n]$ into $\sqrt{n}$ blocks $B_1, \dots, B_{\sqrt{n}}$, each of size $\sqrt{n}$. Define the norm $\|x\| := \max_{k \in [\sqrt{n}]} \left( \sum_{i \in B_k} x_i \right)$. Now suppose $\|\cdot\|'$ is a submodular norm satisfying $\|x\| \leq \|x\|'$ for all $x \in \R_+^n$. We will construct $y$ by starting with the zero vector and iteratively choosing one element $i_k$ of each $B_k$ to activate (set $y_{i_k} = 1$). Clearly $\|y\| = 1$, so we just need to show $\|y\|' \geq \Omega(\sqrt n)$. 
    
    Formally, let $y \up 0 = 0$, and for each $k = 1, \dots,  \frac{\sqrt{n}} 2$, do the following. If $\|y\up{k-1}\|' \geq \frac{\sqrt{n}}2$, we are done and simply choose $y = y\up{k-1}$. Otherwise, notice that $\|y\up{k-1} + \ones_{B_k}\|' \geq \|y\up{k-1} + \ones_{B_k}\| = \sqrt{n}$, so $\|y\up{k-1} + \ones_{B_k}\|' - \|y\up{k-1}\| \geq \frac{\sqrt{n}}2$. By submodularity, there exists some $i_k \in B_k$ such that $\|y\up{k-1} + e_{i_k}\|' \geq \|y \up {k-1}\|' + \frac{1}{2\sqrt{n}}$, for which we set $y\up k := y\up {k-1} + e_{i_k}$ By induction, if we do not terminate early, we have $\|y \up {\sqrt{n}}\|' \geq \frac{n}{2\sqrt{n}} = \frac{\sqrt{n}}{2}$.
\end{proof}

\section{Online Facility Location with Submodular Norms} \label{sec:online-fac}

In this section, we illustrate how submodular norms can be applied to  Online Facility Location. Recall from 
\Cref{sec:introApplications}, in this problem we are given a 
metric space $(\cM, d)$  along with a cost function $f : \cM \to \R_+$. At each time step $i \in [n]$, an adversary produces a new request $x_i \in \cM$, and the algorithm decides whether to  assign $x_i$ to the closest open facility $F_{i-1}$ or to open a new facility $q$ and assign request $x_i$ to  $q$. The goal is to minimize the total facility opening costs plus a given norm $\|\cdot \|$ of the connection costs, i.e., $  \min \sum_{q \in F} f(q) + \|\bd\|,$ 
where $\bd  = (d_1, \dots, d_n) \in \R^{n}_+$ is the vector of connection costs $d_i := d(x_i, F_i)$.

In the case of uniform costs, $f(q) = f$ for all $q$, so the total facility opening cost becomes $f \cdot |F|$.  We use $\bd_{\leq i} = (d_1, \ldots, d_i, 0, \ldots,0)$ to denote the first $i$ coordinates of vector $\bd$.

\subsection{Uniform Costs}

For now, we will focus on the case when facility costs are uniformly $f$.

\begin{theorem}\label{thm:main-unif}
Let $\|\cdot\|$ be a submodular norm, and let $\rho := \rho_{\|\cdot\|}$. For the $\|\cdot\|$ norm online facility location problem with uniform facility costs $f$, there exists a randomized online algorithm that obtains cost at most $O(\log \rho) \cdot |F^*| f + O(1) \cdot \|\bd^*\|$, where $F^*$ and $\bd^*$ are the set of facilities and vector of assignment distances, respectively, given by the optimal offline algorithm.
\end{theorem}

Notice that because our algorithm obtains a constant factor approximation for the assignment costs, we have the following corollary.

\begin{corollary}
There is an $O(\log \rho)$-competitive algorithm for uniform costs online facility location with  symmetric  norms.
\end{corollary}
\begin{proof}
    Given the uniform cost facility location problem with a monotone symmetric norm $\|\cdot\|$, let $F^*$ and $\bd^*$ be the set of facilities and assignments distances given by the optimal offline algorithm. For our online algorithm, we will approximate $\|\cdot\|$ by a submodular norm $\|\cdot\|'$ using \Cref{lem:log-rho-approx}, and then run the algorithm in \Cref{thm:main-unif} on norm $\|\cdot\|'$. Since $\log \rho = \log \rho_{\|\cdot\|} = \Theta(\log \rho_{\|\cdot\|'})$, this algorithm will incur cost at most
    \[
    |F|f + \|\bd\| ~\leq~ |F|f + \|\bd\|' ~\leq~ O(\log \rho) |F^*|f + O(1) \|\bd^*\|' ~\leq~ O(\log \rho) |F^*|f + O(\log \rho) \|\bd^*\|. \qedhere
    \]
\end{proof}

\paragraph*{Proof outline for \Cref{thm:main-unif}} We want to generalize Meyerson's algorithm   beyond $\ell_1$ norms and use submodularity to complete the analysis. 
Meyerson's algorithm constructs a new facility at each demand point $x_i$ with probability ${d(x_i, F_{i-1})}/{f}$, thereby balancing the cost of assigning the  demand against the cost of constructing a new facility. To adapt this algorithm to more general norms, it is natural to construct a new facility at $x_i$ with probability ${\delta_i}/{f}$, where $\delta_i = \|(d_1, \dots, d_{i-1}, d(x_i, F_{i-1}), 0, \dots, 0)\| - \|\bd_{\leq i-1}\|$ is the  marginal cost of assigning $x_i$.
 
Unfortunately, the above natural generalization of Meyerson's algorithm can have an $\Omega(n)$ competitive ratio. For instance, consider the star graph $K_{1, n}$ equipped with the standard unweighted graph distance metric. Suppose that our construction costs are $f = 1$, our submodular norm is the $\ell_\infty$ norm, and the demand points are all the leaves of the star graph. The optimal solution constructs a single facility at the center, yielding a total cost of $1 + \|(1, \dots, 1)\|_\infty = 2$. On the other hand, the suggested algorithm   constructs a facility for every demand point, as $\delta_i = 2$ for each $i \geq 2$, incurring a total cost of $n$.

To get around this issue, 
we will define \emph{auxiliary assignment costs} $\hat{d_i}$ that upper bound the true costs $d_i$. The key modification to the algorithm is that we will use $\hat{d}_i$ instead of $d_i$ for calculating the marginals $\delta_i$. By overestimating the assignment costs that we have incurred, the algorithm underestimates potential marginal costs due to submodularity, making it more inclined to assign demand points instead of constructing new facilities. Moreover, the flexibility of the analysis allows us to show that the increased costs of $\hat{d}_i$ still obtain an $O(\log \rho)$ competitive ratio.

We now present the formal proof.

\begin{proof}[Proof of \Cref{thm:main-unif}]
To formalize the outline above, we will first inductively define our auxiliary cost vector $\hat \bd$ and the marginals $\delta_i$ by
\begin{gather*}
\hat d_i := \min\Big\{d(x_i, F_{i-1}), \;\;\max\{z\geq 0 :f \geq \|(\hat{d_1}, \dots, \hat d_{i-1}, z, 0, \dots, 0)\| -   \|\hat \bd_{\leq i-1}\|\}\Big\} \quad \text{and}\\
\delta_i := \|\hat \bd_{\leq i}\| -   \| \hat\bd_{\leq i-1} \|.
\end{gather*}
Thus, $\hat d_i$ is the assignment distance $d(x_i, F_{i-1})$ capped such that $\delta_i \leq f$. 
For our algorithm, we construct a facility at $x_i$ with probability ${\delta_i}/{f}$, and assign $x_i$ to the nearest facility otherwise. These are well-defined probabilities since $\delta_i \leq f$. To see the upper-bound $d_i \leq \hat d_i$, notice that if $\delta_i < f$, then $\hat d_i = d(x_i, F_{i-1}) \geq d(x_i, F_{i})= d_i$. If $\delta_i = f$, then $d_i = 0$ since a facility is constructed at $x_i$ with probability $1$, so $d_i = 0 \leq \hat d_i$.

Let $\cost(i) := f \cdot \ind_{|F_i| > |F_{i-1}|} + \delta_i$ be the marginal increase in auxiliary cost at step $i$, so $\sum_{i \in [n]} \cost(i) = f \cdot |F| + \|\hat \bd \|$. We will bound separately the cost of demand points that arrive before and after the first nearby facility is constructed. Similar to Meyerson's proof, we have the following bound on costs incurred before a facility is constructed in a given set.

\begin{restatable}{claim}{prefcost}\label{clm:pre-f-cost}
Let $A \subseteq [n]$ be a fixed set of indices, and let $S \subseteq A$ be the subset of indices that arrive before the first facility is constructed at any step in $A$. Then $\E [ \sum_{i \in S} \cost(i) ] \leq 2f$.
\end{restatable}

The proof of this claim is essentially identical to Meyerson's, so we will defer it to \Cref{sec:MissingUnif}.

Now, let us enumerate our offline algorithm's facility set as $F^* = \{c^*_1, \dots, c^*_K\}$, where $K = |F^*|$. Let $C_1^*, \dots, C_K^*$ be the offline clusters, i.e., $C^*_k$ is the set of $i \in [n]$ for which $x_i$ is assigned to $c^*_k$.

Let $r := \frac{\|\bd^*\|}{\|(1, \dots, 1)\|}$. We partition each cluster into rings as $C^*_k = \bigcup_{\ell = 0}^L C_k ^\ell$, where $L = \ceil{\log \rho}$,
\begin{align*}
    C_k^0 &:= \{i \in C_k^* : d(x_i, c^*_k) \leq r\}, \quad \text{and} \\
    C_k^\ell &:= \{i \in C_k^* : 2^{\ell-1}r \leq d(x_i, c^*_k) \leq 2^\ell r\} \quad \textnormal{for } \ell \in \{1, \ldots, L\}.
\end{align*}

Notice that this is a partition since
\[
\max_{i \in C_k^*} d(x_i, C_k^*) \leq \|\bd^*\|_\infty \leq \frac{\|\bd^*\|}{\min_i \|e_i\|} = r\rho \leq 2^L r.
\]

We will analyze the costs incurred by our algorithm on demand points in each ring. Within each ring, we will consider two types of demands separately: \emph{long-distance} demands $\LD_k^\ell$ and \emph{short-distance} demands $\SD_k^\ell$, defined as
\begin{gather*}
    \LD_k^\ell := \{i \in C_k^\ell : d(c^*_k, F_{i-1}) > 2^\ell r\} \quad \text{and} \quad 
    \SD_k^\ell := \{i \in C_k^\ell : d(c^*_k, F_{i-1}) \leq 2^\ell r\}.
\end{gather*}

In other words, long (respectively, short) distance demands arrive before (respectively, after) a facility has been constructed within the outer perimeter of its corresponding ring. We now make the following claims.

\begin{claim}\label{clm:unif-ld}
We have 
$$\E \Big[\sum_{\ell = 0}^L\sum_{k=1}^K\sum_{i \in \LD_k^\ell} \cost(i) \Big] \leq 2(L+1)Kf.$$
\end{claim} 

\begin{claim}\label{clm:unif-sd}
We have
$$\E \Big[\sum_{\ell = 0}^L \sum_{k=1}^K\sum_{i \in \SD_k^\ell} \cost(i) \Big] \leq 8\|\bd^*\|.$$
\end{claim}

These claims together give \Cref{thm:main-unif}, since $L = O(\log \rho)$.

Notice that \Cref{clm:unif-ld} follows immediately from \Cref{clm:pre-f-cost} since long-distance demands must arrive before a facility is constructed   in $C \up j _\ell$. This implies $\E \left[\sum_{i \in \LD_k^\ell} \cost(i)\right]  \leq 2f$ for each $\ell$ and $k$. 

To show \Cref{clm:unif-sd}, let $\SD := \bigcup_{\ell=0}^L \bigcup_{k= 1}^K \SD_k^\ell$. We seek to show that $\sum_{i \in \SD} \cost(i) \leq 8\|\bd^*\|$. Notice that if $i \in \SD_k ^0$ for some $k$, then we have  $\hat d_i \leq d(x_i, F_{i-1}) \leq d(x_i, c^*_k) + d(c^*_k, F_{i-1}) \leq d^*_i + r$. Similarly, if $i \in \SD_k^\ell$ for some $k$ and $\ell \neq 0$, we have $\hat d_i \leq 3d^*_i$. Thus, we can say that $\hat d_i \leq 3d^*_i + r$ for all $i \in \SD$. This gives
\[
\E \Big[\sum_{i \in \SD} \cost(i)\Big] ~\leq~ \E \Big[\sum_{i \in \SD} 2 \delta_i\Big] ~\leq~ 2\E \Big[\|\hat \bd_{\SD}\|\Big] ~\leq~ 2\|r \ones_{\leq n} + 3\bd^*\| ~=~ 2r\|\ones_{\leq n}\| + 6\|\bd^*\| ~\leq~ 8\|\bd^*\|,
\]
where the second inequality comes from the submodular property, the third comes from norm monotonicity, and the last inequality is by choice of $r$. In particular, for the second inequality, we crucially use the submodularity of the norm to say
\[
\sum_{i \in \SD} \delta_i ~=~ \sum_{i \in \SD} \left( \|\hat \bd_{\leq i}\| - \|\hat \bd_{\leq i-1}\| \right) ~\leq~ \sum_{i \in \SD} \left( \|\hat \bd_{\SD \cap [i]}\| - \|\hat \bd_{\SD \cap [i-1]}\| \right) ~=~ \|\hat \bd_{\SD}\|. \qedhere
\]
\end{proof}

\subsection{Non-Uniform Costs} \label{sec:nonUnif}
In this section, we argue how one could develop the ideas of the previous section further to extend the algorithm to Online Facility Location with different opening costs across facilities. We will show how to modify Meyerson's algorithm for non-uniform costs in a similar manner to the uniform cost setting, but in a way that also handles new challenges that arise. We motivate and describe the new algorithm, but to avoid clutter, we defer most of the proof details to \Cref{sec:MissingNonUnif}.

\nonUnif*

Recall that previously, the algorithm's only choice at step $i$ was whether to construct a facility at $x_i$ or not. In the non-uniform setting, it might not be feasible to only ever construct at $x_i$, since the cost of opening there might be prohibitively high. Meanwhile a nearby location might have a much lower cost. Instead, the algorithm must consider all possible cost levels at which it could construct and how far facilities at different cost levels are from $x_i$.

First, let us recap how Meyerson's algorithm handles this in the $\ell_1$ norm setting. By losing at most a factor of $2$, we can assume that all opening costs are in some set $\{f_1, \dots, f_m\}$, where each $f_j$ is a power of 2. For each $j \in [m]$, we define the $W_i \up j$ to be the set of facilities which at step $i$ are open or have opening cost at most $f_j$:
$$
W_i\up j := F_{i-1} \cup \{x \in \cM: f(x) \leq f_j\}.
$$
Additionally, $W_i\up 0 := F_{i-1}$. Now Meyerson's algorithm \cite{Mey01} will, for each $j \in [m]$, open a facility at the nearest location in $W_j$ with probability $\frac{d(x_i, W_i\up {j-1}) - d(x_i, W_i \up j)}{f_j}$, capped at 1, then assign $x_i$ to the nearest open facility. As in the uniform case, we can see that the expected facility opening cost incurred is $d(x_i, W_i \up 0) = d(x_i, F_{i-1})$. This allows us to again consider a ``long-distance'' and ``short-distance'' phase within each ring of each optimal cluster.

To adapt this algorithm to the submodular norm setting, the idea will again be to consider construction probabilities, not based on the distances $d(x_i, W_i \up j)$, but instead based on the marginal increase in an auxiliary assignment cost $\|\hat \bd\|$. However, the definition of $\hat \bd$ is not as straightforward as  in the uniform cost setting. In particular, we need to satisfy $d_i \leq \hat d_i$ for each $i$ deterministically. However, without carefully controlling dependencies between facility construction (for instance, by constructing independently for each $j$), it may be possible that $d_i = d(x_i, F_{i-1})$ if no facilities are constructed. This would force $\hat d_i$ to be too large to serve as a useful upper bound.

To avoid this, we instead construct facilities by sampling a single cost level $f_j$ at which to build a facility. 
The probability of sampling $f_j$ is given by a version of the probabilities used above. 
Suppose these probabilities sum to a value greater than 1. In that case, we crucially limit the probabilities for smaller $f_j$ (corresponding to a facility at a greater distance) until the total probability is capped at 1: We name these probabilities $p_i\up{j}$'s (\Cref{def:cap}). This ensures that the true assignment distance $d_i$ is never too large, which allows us to pick a conveniently small upper bound $\hat d_i$.

\begin{algorithm}
\caption{Non-Uniform Submodular Online Facility Location}\label{alg:nunif-ofc}
\KwData{Metric space $(\cM, d)$ and online requests $x_1, \ldots, x_n$}
\KwResult{Opened facilities set $F$ and assignment of requests to facilities (at the time of arrival)}
\For{request $x_i$, $i = 1, \ldots, n$}{
    Sample $j \in \{0, \dots, m\}$ according to the distribution given by $p_i\up j$\;
    Assign $x_i$ to the nearest location in $W_i\up j$, constructing a facility there if necessary\;
}
\end{algorithm}

Our analysis proceeds by considering the partition of arrival indices into optimal clusters $[n] = \bigcup_{k \in [K]} C_k$, each with center $c^*_k$. Similarly to the uniform cost setting, we partition each cluster into rings as $C^*_k = \bigcup_{\ell = 0}^L C_k ^\ell$, where $L = \ceil{\log \rho}$ and
\begin{align*}
    C_k^0 &:= \{i \in C_k^* : d(x_i, c^*_k) \leq r\},\\
    C_k^\ell &:= \{i \in C_k^* : 2^{\ell-1}r \leq d(x_i, c^*_k) \leq 2^\ell r\}, \quad \textnormal{for } \ell \in [L].
\end{align*}

For each ring $C_k^\ell$, we divide the analysis into two main stages: the \emph{short-distance} and \emph{long-distance} stages. We leave the formal definition of these stages for the proof details, but intuitively, each demand is considered a long-distance demand until a facility is constructed at a distance which is a constant multiple of the radius of $C_k^\ell$. After this happens, subsequent demands in the ring are considered short-distance demands.
In each of these stages, the algorithm obtains the following.

\begin{restatable}[Short-Distance Stage]{lemma}{lemsd}\label{lem:sd}
    The expected cost incurred by the algorithm in the short-distance stage is
    \begin{align*}
        \E\left[\ALG_{\SD}\right] = \E \Big[\sum_{i \in \SD} \cost(i)\Big] \leq 36 \norm{\bd^*}.
    \end{align*}
\end{restatable}


\begin{restatable}[Long-Distance Stage]{lemma}{lemld}\label{lem:ld}
    The expected cost incurred  in the long-distance stage is
    \begin{align*}
        \E\Big[\ALG_{\LD}\Big] \leq 48 (\log \rho + 1) \cdot \sum_{k \in [K]} f(c^*_k).
    \end{align*}
\end{restatable}

It is now easy to see that  \Cref{thm:non-unif} directly follows by combining \Cref{lem:sd} and \Cref{lem:ld}. The proofs of these lemmas, along with a more formal definition of the probabilities $p_i \up j$, are given in \cref{sec:MissingNonUnif}.


\section{Adaptivity Gaps for Stochastic Probing} \label{sec:stochProbing}

In this section,  we  use submodular norms to prove small adaptivity gaps for the stochastic probing problem.
Recall the   stochastic probing problem from \Cref{sec:introApplications}:  Given $n$ independent random variables $X = (X_1, \dots, X_n) \in \R_+$, a downward closed set family $\cF \subseteq 2^{[n]}$, and a monotone objective $f : \R^n_+ \to \R_+$, the stochastic probing problem $(X, \cF, f)$ is to open a feasible set $S \in \cF$ of variables  to maximize $f(X_S)$. 

Denote by $\adap(X, \cF, f)$ the maximum expected objective achievable by an adaptive algorithm, i.e., one which selects elements of $S$ one at a time, and may change its strategy based on its observations of the selected variables. We denote by $\na(X, \cF, f)$ the maximum expected objective by a non-adaptive algorithm, i.e., $\na(X, \cF, f) := \max_{S \in \cF} \E[f(X_S)]$.

\begin{theorem} \label{thm:adapGapSubmod}
If $f$ is a submodular norm, then $\adap(X, \cF, f)  \leq 2 \cdot \na(X, \cF, f)$.
\end{theorem}

The following result for symmetric norms is an immediate corollary due to \cref{thm:symTosubmod}.

\adapGapSymmetric*

\begin{proof}[Proof of \cref{thm:adapGapSubmod}] We follow the same proof approach as in \cite{BSZ-Random19}. 
Consider an adaptive algorithm $\adap$, and a non-adaptive algorithm $\alg$ which selects each $S \in \cF$ with the same probabilities as $\adap$, only non-adaptively. We will show by induction on $n$ that $\adap$ achieves an expected objective at most twice that of $\alg$. This is trivially true for $n=1$, so we only need to show the inductive step.

We can compare the performance of these two algorithms by coupling their actions. Let's say $\adap$ runs on random variables $X = (X_1, \dots, X_n)$, and let $S \in \cF$ be the (random) set of adaptively chosen variables. We can say $\alg$ runs on variables $Y = (Y_1, \dots, Y_n)$, i.i.d. copies of $X$, by choosing the same set $S$ as $\adap$. Without loss of generality, say that $\adap$ starts by selecting $X_1$. Since $1 \in S$ deterministically, we have that $\adap$ achieves reward
\begin{align*}
    \E [f(X_S)] &= \E [f(X_1)]  + \E\big[f_{X_1}(X_{S \setminus \{1\}} )\mid X_1\big]\\
    &\leq \E [f(X_1 \vee Y_1)]  + \E\big[f_{X_1 \vee Y_1}(X_{S \setminus \{1\}} )\mid X_1\big]\\
    &\leq 2\E [f(Y_1)]  + \E\big[f_{X_1 \vee Y_1}(X_{S \setminus \{1\}} )\mid X_1, Y_1\big],
\end{align*}
where $f_{x}(Z) := f(x, Z) - f(x, 0)$ for $Z = (Z_2, \dots, Z_n) \in \R^{n-1}_+$. Notice that, since $f$ is submodular, we have $f_x$ is submodular and decreasing in $x$ for all $x \in \R_+$.

Now, notice that $\alg$ achieves reward
\begin{align*}
    \E [f(Y_S)] &= \E[f(Y_1)] + \E\big[f_{Y_1}(Y_{S \setminus \{1\}} )\mid X_1\big]\\
    &\geq \E [f(Y_1)] + \E\big[f_{X_1 \vee Y_1}(Y_{S \setminus \{1\}} )\mid X_1, Y_1\big].
\end{align*}

Notice that given $X_1$, the set $S \setminus \{1\} \in \cF |_{-1}$ is adaptively chosen among the variables $X_2, \dots, X_n$ by $\adap$. Thus, by induction we can say $\E\big[f_{X_1 \vee Y_1}(X_{S \setminus \{1\}} )\mid X_1, Y_1\big] \leq 2 \cdot \E\big[f_{X_1 \vee Y_1}(Y_{S \setminus \{1\}} )\mid X_1, Y_1\big]$. Combining this with the above inequalities gives
\[
\E [f(X_S)] \leq 2\E [f(Y_1)]  + 2\E\big[f_{X_1 \vee Y_1}(Y_{S \setminus \{1\}} )\mid X_1, Y_1\big] \leq 2 \E [f(Y_S)]. \qedhere
\]
\end{proof}

\section{Conclusion}
This paper introduces the concept of submodular norms and demonstrates their application in proving the efficiency of optimization problems beyond traditional $\ell_p$ objectives. We provide examples showcasing the utility of submodular norms in various scenarios. Specifically, we establish bounds on the competitive ratio of online facility location problems and the adaptivity gap of stochastic probing techniques when using symmetric norm objectives. These bounds crucially depend on the norm parameter $\rho$, and are approximately tight in the case of facility location. We also obtain an alternative algorithm for certain generalized load balancing settings using our techniques. There are several natural directions for future work:

(i) \emph{General Monotone Norms:} 
We have shown a logarithmic competitive ratio and adaptivity gap for online facility location and stochastic probing, respectively,  when the objective is a symmetric norm or approximately a submodular norm. However, it remains open whether poly-logarithm bounds exist for either problem when the norm can be an arbitrary monotone norm.

(ii) \emph{Symmetric Norm Stochastic Probing.}
The logarithmic factor we get in our adaptivity gap bound for symmetric norm stochastic probing comes from the loss in approximating a symmetric norm by a submodular norm. However, it is not clear if such a loss is necessary. It would be interesting to determine if the true adaptivity gap is sub-logarithmic or even a constant.

(iii) \emph{Parameter $\rho$.}
Similar to online facility location, there are other optimization problems (e.g., online fractional set cover) which are known to have differing performance guarantees for $\ell_1$ and $\ell_\infty$ objectives. We hypothesize that for such problems with symmetric norm objectives, the parameter $\rho$ could provide a way of interpolating between $\ell_1$ and $\ell_\infty$.

{\small
\newcommand{\etalchar}[1]{$^{#1}$}

}

\appendix

\section{Omitted Proofs from \Cref{sec:submod}} \label{sec:missingSubmod}

\subsection{Properties of Continuous Submodularity} \label{sec:contSubmod}
 
As with submodular set functions, there are many equivalent definitions for continuous submodularity which may be helpful in different settings. These are folklore properties, but we prove them for completeness.

\begin{lemma}\label{lem:submod-tfae}
    Let $f : \R^d_+ \to \R_+$. The following are equivalent.
    \begin{enumerate}
        \item $f$ is continuously submodular.
        \item For all $x, y, z \in \R^d_+$ with $\Supp(y) \cap \Supp(z) = \emptyset$, we have $$f(x) + f(x + y + z) \leq f(x + y) + f(x + z).$$
        \item \label{eq:contSubmodDR} For all $x,y \in \R^d_+$ with $x \leq w$, and $i \in [d]$ such that $x_i = w_i$, and $a \geq 0$, we have
        $$
        f(w + ae_i) - f(w) \leq f(x + ae_i) - f(x).
        $$
        \item For all $x \in \R^d_+$ and $a,b \geq 0$ and distinct $i,j \in [d]$, we have
        $$
        f(x) + f(x + ae_i + be_j) \leq f(x + ae_i) + f(x + be_j).
        $$
    \end{enumerate}
\end{lemma}


\begin{proof}[Proof of \Cref{lem:submod-tfae}]
($1 \iff 2$) Let $x, y, z \in \R^d_+$ with $y \perp z$. Notice that for non-negative vectors $y,z$, orthogonality implies that they have disjoint support. Hence, $(x + y) \vee (x + z) = x$ and $(x + y) \wedge (x + z) = x + y + z$. Then $2$ follows from the definition of continuous submodularity. Likewise, if $f$ satisfies condition $2$, then applying condition 2 with $x' := x \wedge y$, $y' := y_{\{i : y_i > x_i\}}$, and $z' := x_{\{i : x_i > y_i\}}$ gives continuous submodularity.

($2 \implies 3$) Simply take $y = w - x$ and $z = ae_i$.

($3 \implies 4$) Simply take $w = a + be_j$.

($4 \implies 2$) Let $x\up{i,j} := x + y_{< i} + z_{< j}$. Consider the sum
\begin{align*}
    & f(x + y) + f(x + z) -f(x) - f(x + y + z) \\
    = &\sum_{i,j \in [d]} \left[f(x\up{i,j} + y_ie_i) + f(z_i\up{i,j} + be_j)-f(x\up{i,j}) - f(x\up{i,j} + y_ie_i + z_je_j) \right].
\end{align*}
By condition 4, every term in the RHS sum is non-negative, so the LHS is non-negative as well.
\end{proof}

A commonly studied variant of continuous submodularity is  DR-submodularity \cite{BianB019, FeldmanK20, NiazadehRW18}: a function $f : \R^d_+ \to \R_+$ is \emph{DR-submodular} if it satisfies the stronger condition that for all $x, w \in \R^d_+$ with $x \leq w$, $i \in [d]$, and $a \geq 0$, we have
 $ f(w + ae_i) - f(w) \leq f(x + ae_i) - f(x).$ 
In other words, $f$ satisfies condition \ref{eq:contSubmodDR} of \Cref{lem:submod-tfae} even where $w_i \neq x_i$. However,   the only DR-submodular norm is the $\ell_1$-norm. 

\begin{lemma}
Any DR-submodular norm is equivalent to  $\ell_1$ up to rescaling the coordinates. 
\end{lemma}
\begin{proof}
Suppose $\|\cdot\| : \R^d_+ \to \R$ is a DR-submodular norm with $\|e_i\| = 1$ for each $i \in [d]$. Clearly, $\|x\| \leq \sum x_i = \|x\|_1$ by triangle inequality. Suppose that $\|x\| < \|x\|_1$ for some $x \in \R^d_+$. Then for some $i \in [d]$, we have $\|x_{\leq i}\| - \|x_{\leq i-1}\| \leq x_i - \varepsilon$, where $\varepsilon > 0$. By DR submodularity, this means $\|x_{\leq i-1} + kx_i e_i\| \leq k(x_i - \varepsilon)$ for all $k \in \N$. However, with the continuity of norms, this gives
\[
x_i - \varepsilon \geq \lim_{k \to \infty} \frac{\|x_{\leq i-1} + kx_i e_i\|}{k} = \lim_{k \to \infty} \norm{\frac{1}{k} \cdot x_{\leq i-1} + x_i e_i} = \|x_i e_i \| = x_i,
\]
which is a contradiction, so we have $\|x\| = \|x\|_1$.
\end{proof}

\subsection{Alternative Proof of \cref{lem:log-rho-approx}}
We reprove the lemma using the alternative norm definition 
$\|x\|' := 2\sum_j \inprod{b_j, x^\downarrow}$, where $b_j := \frac{\|\ones_{\leq m_j}\|}{m_j} \cdot \ones_{\leq m_j}$ for each 
$b_j \in \{0, \dots, \floor{\log \rho}\}$.

Let $a_j$ for each $j$ be defined as before. Notice that $a_j$ must majorize $b_j$, since $\inprod{a_j, \ones_{\leq m_j}} = \|\ones_{\leq m_j}\| = \inprod{b_j, \ones_{\leq m_j}}$, and the first $m_j$ coordinates of $b_j$ are identical. Schur-convexity then tells us that for any $x \in \R_+^n$, we have $\inprod{b_j, x^\downarrow} \leq \inprod{a_j, x^\downarrow}$. Thus we have
$$
\frac{1}{2(\floor{\log \rho} + 1)}\|x\|' \leq \max_{j} \inprod{b_j, x^\downarrow} \leq \max_{j} \inprod{a_j, x^\downarrow} \leq \max_{a \in \cA} \inprod{a, x^\downarrow} = \|x\|.
$$

Next, we write $x^\downarrow = \sum_{k \in [n]} \lambda_k \ones_{\leq k}$ for some $\lambda_k \geq 0$. Since $\inprod{b_j, \ones_{\leq m_j}} = \|\ones_{\leq m_j}\|$, just like in the original proof we have
$$
\|x\|' = 2\sum_j \sum_k \lambda_k \inprod{b_j, \ones_{\leq k}} 
\geq 2\sum_k \lambda_k \max_{j} \inprod{b_j, \ones_{\leq k}} 
\geq \sum_k \lambda_k \|\ones_{\leq k}\| \geq \|x\|.
$$

\section{Omitted Proofs from \Cref{sec:online-fac}} 

\subsection{Uniform Costs} \label{sec:MissingUnif}

\prefcost*

\begin{proof}
Consider the following game: For $i \in A$, a player is shown $\delta_{i}$, and has to option to pay a cost of $\delta_{i}$ to play a lottery, which has a $\frac{\delta_{i}}{f}$ chance of giving reward $f$. Since the expected reward of playing the lottery is exactly the cost, the player is indifferent to playing at each step. This means any strategy for the player has zero expected reward. In particular, the strategy of playing the lottery only until the first win has an expected reward of 0.

Let $R$ be the total lottery winnings of this strategy and $C$ be the total cost of playing. We have that $\E[R - C] = 0$, and since at most one lottery is won, $\E[R] \leq f$. Thus, $\E[R + C] = 2 \E[R] \leq 2f$. But $R + C$ has exactly the distribution of $\sum_{i \in S} \cost(i)$, which gives the desired result.
\end{proof}

\subsection{Non-uniform Costs} \label{sec:MissingNonUnif}

We now introduce the notation needed to prove that Algorithm~\ref{alg:nunif-ofc} shows  \Cref{thm:non-unif}.  Let us recall that $f : X \to \R_+$ is the cost function of opening a facility. First, we assume without loss of generality that $f(x)$ is a power of $2$ for each $x \in X$.\footnote{That is, by rounding costs down to powers of $2$, we lose only a factor of 2.} Let $f\up 1 \leq \dots \leq f\up m$ be the distinct cost levels, so $2 f\up i \leq f\up {i+1}$. Additionally, let $f\up 0 := 0$ for completeness. As before, let $F_i$ denote the set of facilities that have been opened after the arrival of $x_i$, before the arrival of $x_{i+1}$.

\begin{definition}\label{def:cap}
    For each step $i \in [n]$ and cost level $j \in \{0, \dots, m\}$, let us define

\begin{enumerate}
    \item $W_i \up j := F_{i-1} \cup \{x \in X: f(x) \leq f\up j\}$ to be the set of locations which are open or have an opening cost at most $f\up j$;
    \item $\hat d_i \up j := \min\{d(x_i, W_i \up j), \tau_i\}$ to be the capped value of $d_i \up j$ (where cap $\tau_i$ is defined in \Cref{pt5});
    \item $\hat d_i := \hat d_i \up 0 = \min\{d(x_i, F_{i-1}), \tau_i\}$ for simplicity.
    \item $\delta_i \up j := \frac{\hat d_i \up j}{\hat d_i \up 0}\left(\|\hat \bd_{\leq i} \up 0\| - \|\hat \bd_{\leq i-1} \up 0\|\right)$ to be the fraction of marginal increase in assignment cost we attribute to cost levels $\leq j$;
    \item $p_i \up j := \frac{\delta_i \up {j-1} - \delta_i \up j}{f\up j} 
    $ for $j \geq 1$ to be the assigned probability of opening a facility in $W_i\up j$, and $p_i \up 0 := 1 - \sum_{j=1}^m p_i \up j$;
    \item $\tau_i := \arg\max\{\tau \in \R_{\geq 0} \cup \{+\infty\} \mid \sum_{j=1}^m p_i \up j \leq 1\}$ to be the cap value, i.e., the largest nonnegative cap such that $\sum_{j=1}^m p_i \up j \leq 1$. This exists as each $p_i \up j$ is monotone decreasing in $\tau_i$. \label{pt5}
\end{enumerate}
\end{definition}

Given the above definitions, we notice that since $d_i \up j$ is decreasing in $j$, this means for some $j$ we have
    \begin{align*}
        0 = \hat d_i\up m \leq \hat d_i \up {m-1} \leq \dots \leq \hat d_i \up {j+1} \leq \tau_i = \hat d_i \up j = \dots = \hat d_i \up 0.
    \end{align*}
We shall prove the following theorem as discussed in \Cref{sec:nonUnif}.

\nonUnif*

To prove this theorem, we will separately the so-called short distance demands $\SD_k^\ell$ and long distance demands $\LD_k^\ell$ in each ring $C_k^\ell$. Formally, we define
\begin{align*}
    \LD_k^\ell &:= \{i \in C_k^\ell : \hat d_i\up 0 > (\lambda + 1)2^\ell r\},&
    \LD := \bigcup_{k = 1}^K \bigcup_{\ell = 0}^L \LD_k^\ell,\\
    \SD_k^\ell &:= \{i \in C_k^\ell : \hat d_i\up 0 \leq (\lambda + 1)2^\ell r\},&
    \SD := \bigcup_{k = 1}^K \bigcup_{\ell = 0}^L \SD_k^\ell.
\end{align*}

\subsubsection{Short-distance stage}

\lemsd*

\begin{proof}
    Let us fix a set $C_k^\ell$. If $\ell = 0$, then for all $i \in \SD_k^\ell$, we have that $\hat d_i \leq (\lambda + 1)r$. If $\ell > 0$, we still have that $\hat d_i \leq (\lambda + 1)2^\ell r \leq  2(\lambda + 1)d^*_i$. Summing up overall demands arriving in the short distance stage we have,
    \begin{align*}
        \E\left[\ALG_{\SD}\right] &\leq \sum_{i \in \SD} \E[\cost(i)] \leq \sum_{i \in \SD} 2\delta^{(0)}_i = \sum_{i \in \SD : d^*_i \leq r} 2\delta^{(0)}_i + \sum_{i \in \SD : d^*_i > r} 2\delta^{(0)}_i\\
        &\leq 2\cdot\|(\hat d_i)_{i \in \SD : d^*_i \leq r}\| + 2\cdot\|(\hat d_i)_{i \in \SD : d^*_i > r}\|\\
        &\leq 2(\lambda + 1)\cdot r \cdot \|(1 \ldots 1)\| + 4(\lambda + 1)\cdot\|\bd^*\|\\
        &\leq 6(\lambda + 1)\cdot\|\bd^*\|.
    \end{align*}
    Here, the second inequality comes from the fact that we need to account for the facility opening cost as well as the connection cost. Moreover, the third inequality holds by norm submodularity, the fourth by what was argued earlier on distances, and the last by definition of $r$. The lemma then follows from choosing $\lambda = 5$, which is needed for the proof of \cref{lem:ld}.
\end{proof}

\subsubsection{Long-distance stage} 

\lemld*

\begin{proof}
    Let us fix a cluster ring $C_k^\ell$, and let $j_k^*$ be defined such that $f(c_k^*) = f\up {j_k^*}$. Denote by $\gamma_i^{(j)} := d(c^*_k, W^{(j)}_i)$, the distance at step $i$ between the cluster center and a facility whose opening cost is at most $f^{(j)}$. We denote by $\cE^{(j)}_\ell$ the event that a facility is opened within a $\gamma_0^{(j)} + 2^{\ell+1} r$ distance from optimal center $c^*_k$. It is easy to see that such an event occurs whenever the algorithm constructs a facility of cost $f^{(j)}$ or higher for a demand in $C_k^\ell$. We now analyze the expected cost accumulated by the algorithm before, and after $\cE^{(j)}_\ell$ has occurred. We denote by $t^{(j)}_\ell$ the time of event $\cE^{(j)}_\ell$ occurrence.

    Before $\cE^{(j)}_\ell$ has occurred, we have $\sum_{i \leq t^{(j)}_\ell}\E[\cost^{(j)}(i)] \leq 2f^{(j)}$, by the same reasoning as \cref{clm:pre-f-cost}. Hence, we have that the total cost all levels $j \leq j_k^*$ before event $\cE^{(j)}_\ell$ is
    \begin{align*}
        \sum_{j \leq j^*_k} \sum_{\substack{i \in \LD_k^\ell \\ i \leq t^{(j)}_\ell}} \E[\cost^{(j)}(i)] \leq 2\sum_{j \leq j^*_k} f^{(j)} \leq 4f^{(j^*_k)}.
    \end{align*}
    We seek to demonstrate that these costs make up a constant fraction of all costs during the long-distance stage. Notice that by re-indexing, we can write
    $$
    \sum_{j \leq j^*_k} \sum_{\substack{i \in \LD_k^\ell \\ i < t^{(j)}_\ell}} \E[\cost^{(j)}(i)] = \sum_{j=0}^{j_k^*-1} \sum_{\substack{i \in \LD_k^\ell \\ t_\ell\up {j} < i \leq  t_\ell\up{j+1}}} \sum_{j'=j+1}^{j_k^*} \E[\cost^{(j')}(i)] = \E\sum_{j=0}^{j_k^*-1} \sum_{\substack{i \in \LD_k^\ell \\ t_\ell\up {j} < i \leq  t_\ell\up{j+1}}} 2(\delta_i \up j - \delta_i \up {j_k^*}),
    $$
    i.e., these are also the costs that occur in the range $\{j+1, \dots, j_k^*\}$, during each period between event $\cE^{(j)}_\ell$ and $\cE^{(j+1)}_\ell$. In particular, we will show for each term in the sum, $\nicefrac{(\delta^{(j)}_i - \delta_i\up {j_k^*})}{\delta^{(0)}_i} \geq \nicefrac{\lambda - 4}{\lambda + 1}$, so these costs comprise a constant fraction of the total expected cost $2 \delta_i \up 0$ at each step $i$. 
    
    We start with the simple observation that
    \begin{align*}
        \frac{\delta^{(j)}_i - \delta_i \up {j_k^*}}{\delta^{(0)}_i} = \frac{\hat d^{(j)}_i - \hat d_i\up {j_k^*}}{\hat d^{(0)}_i} \geq \frac{d(x_i, W_i\up j)}{d(x_i, F_{i-1})} - \frac{d(x_i, c^*_k)}{\hat d_i \up 0},
    \end{align*}
    by definition and subadditivity. We now proceed with bounding each term. Since $\cE^{(j)}_\ell$ has occurred, but we are still in the long distance stage, we have
    \begin{align*}
        (\lambda + 1)2^\ell r < \hat d_i \up 0 \leq d(x_i, F_{i-1}) \leq d(x_i, c^*_k) + d(c^*_k, F_{i-1}) \leq \gamma_i^{(j)} + 3\cdot 2^\ell r. 
    \end{align*}
    This implies both $d(x_i, F_{i-1}) \leq \gamma_i^{(j)} + 3\cdot 2^\ell r$ and $\gamma_i \up j \geq (\lambda - 2)2^\ell r$. Additionally, we have $d(x_i, c_k^*) \leq 2^\ell r$, and by triangle inequality, we have that $\gamma_i^{(j)} \leq d(x_i, W_i\up j) + 2^{\ell} r$. 
    
    Altogether, we get
    \begin{align*}
        \frac{d(x_i, W_i\up j)}{d(x_i, F_{i-1})} - \frac{d(x_i, c^*_k)}{\hat d_i \up 0}\geq \frac{\gamma_i^{(j)} - 2^\ell r}{\gamma_i^{(j)} + 3\cdot 2^{\ell+1} r} - \frac{2^\ell r}{(\lambda + 1)2^\ell r}\geq \frac{\lambda - 4}{\lambda + 1},
    \end{align*}
    as desired. Thus, the total cost of points in $\LD_k ^\ell$ is bounded as follows:
    \begin{align*}
        \sum_{i \in \LD_k^\ell} \E [\cost (i)] ~=~
        \sum_{j=0}^{j_k^*-1} \sum_{\substack{i \in \LD_k^\ell \\ t_\ell\up {j} < i \leq  t_\ell\up{j+1}}} \E [\cost (i)] 
        &~=~\E\sum_{j=0}^{j_k^*-1} \sum_{\substack{i \in \LD_k^\ell \\ t_\ell\up {j} < i \leq  t_\ell\up{j+1}}} 2\delta_i \up 0\\
        &\leq 2 \frac{\lambda + 1}{\lambda - 4} \cdot \E\sum_{j=0}^{j_k^*-1} \sum_{\substack{i \in \LD_k^\ell \\ t_\ell\up {j} < i \leq  t_\ell\up{j+1}}} (\delta_i \up j - \delta_i \up {j_k^*})\\
        &\leq 8 \frac{\lambda + 1}{\lambda - 4} f\up{j^*_k}.
    \end{align*}
    We now sum across all concentric rings across all $K$ optimal clusters, to obtain that 
    \begin{align*}
        \sum_{i \in \LD} \E [\cost(i)] = \sum_{k = 1}^K \sum_{\ell = 0}^L \sum_{i \in \LD_k^\ell} \E[\cost(i)] \leq 8L \left(\frac{\lambda + 1}{\lambda - 4}\right) \cdot \sum_{k \in [K]} f^{(j^*_k)},
    \end{align*}
    and the claim follows from choosing $\lambda = 5$.
\end{proof}

\subsection{Lower Bound}\label{sec:ofl-lower-bound}
\begin{theorem}\label{thm:ofl-lower-bound}
    For any monotone norm $\|\cdot\|$, there exists a uniform-cost OFL problem with norm $\|\cdot\|$ such that any online algorithm only achieves $\Omega\left(\frac{\log \sigma}{\log \log \sigma}\right)$ competitive ratio, where $\sigma = \frac{\|\ones_{\leq n}\|}{\max_i \|e_i\|}$. Notice that for symmetric norms, $\sigma= \rho$.
\end{theorem}
\begin{proof}
    We may assume $\max_i \|e_i\| = 1$ (otherwise we rescale costs), so $\|\ones_{\leq n}\| = \sigma$. Let $k$ be the largest integer such that $k^k \leq \sigma$, so we have $k = \Theta\left(\frac{\log \sigma}{\log \log \sigma}\right)$. Assume $k \geq 2$.
    
    Now, let $G = (V, E)$ be a complete $N$-ary tree with height $k$, where $N$ is sufficiently large (intuitively, think of $N$ as infinite). For $j = 0, \dots k-1$, each downwards edge from a node at depth level $j$ will have length $k^{-j}$. We also define the facility opening cost to be $k$.
    
    Let $0 \leq m_0 \leq m_1 \leq \dots \leq m_k$ be defined such that $m_j$ is the least positive integer with $\|\ones_{\leq m_j}\| \geq k^j$. Notice that this implies $\|\ones_{\leq m_j}\| \leq \|\ones_{\leq m_j - 1}\| + \|e_{m_j}\| < k^j + 1$. 
    
    Our adversary will supply the demand locations as follows. First, they will choose a random path $v_0 v_1 \dots v_k$ from the root $v_0$ to a leaf $v_k$. Then, for $j=0, \dots, k$, the adversary supply $v_j$ as a demand repeated $m_j - m_{j-1}$ times ($m_0$ times for $j = 0$).
    
    In the offline setting, one may simply place a single facility at $v_k$ and assign all demands to it. This gives
    $$
    \OPT = k + \|\bd^*\| \leq k + \sum_{j=0}^{k-1} k^{-j} \|\ones_{\leq {m_j}}\| \leq k + \sum_{j=0}^{k-1} k^{-j} (k^j + 1) = O(k).
    $$
    
    In the online setting, we will show that no algorithm can achieve an expected cost of less than $\Omega(k^2)$.
    
    Notice that any online algorithm, upon receiving a demand at $v_j$, should only consider the options of allocating the demand or constructing a facility at $v_j$. Constructing a facility anywhere else is strictly disadvantageous, as there is a negligible probability (by choice of $N$) that the chosen location is in the subtree rooted at $v_{j+1}$. Thus, after the algorithm is complete, it will have constructed a set of facilities $F \subseteq \{v_0, \dots, v_k\}$, and each demand will be allocated to the most recently constructed facility above it. Let $\bd$ be the vector of allocation distances.
    
    If there is some $j \geq 2$ such that $v_j, v_{j-1} \not \in F$, then notice that every demand at $v_j$ will have allocation distance at least $k^{-j+2}$. Thus, we have
    $$
    \|\bd\| \geq k^{-j+2} \|\ones_{\leq m_j} - \ones_{\leq m_{j-1}}\| \geq k^{-j+2} (\|\ones_{\leq m_j}\| - \|\ones_{\leq m_{j-1}}\|) \geq k^2 - k - k^{-j+2} = \Omega(k^2)
    $$
    However, if no such $j$ exists, then $|F| \geq k/2$, so construction costs are at least $k^2/2 = \Omega(k^2)$.
\end{proof}

\begin{corollary}
    In the case of a symmetric norm $\|\cdot\|$, our lower bound becomes $\Omega\left( \frac{\log \rho}{\log \log \rho}\right)$ as $\rho = \sigma = \frac{\|\ones_{\leq n}\|}{\|e_1\|}$.
\end{corollary}

\end{document}